%% file: SellingPrivacy.tex
\documentclass[11pt]{article}
\usepackage{amsmath}
\usepackage{amssymb}
\usepackage{amsthm}
\usepackage[margin=1in]{geometry}
\usepackage{times}
\usepackage{verbatim}
\usepackage{authblk}
\usepackage[ruled]{algorithm2e}
\newtheorem{theorem}{Theorem}[section]
\newtheorem{lemma}[theorem]{Lemma}
\newtheorem{asm}{Assumption}

\newcommand{\Paren}[1]{\ensuremath{\left ( #1 \right )}} 
\title{Approximately Optimal Auctions for Selling Privacy \\when Costs are Correlated with Data\footnote{E-mail: \{lkf,yuhanlyu\}@cs.dartmouth.edu. Partially supported by NSF grants CCF-0728869 and CCF-1016778.}}

\author{Lisa Fleischer}
\author{Yu-Han Lyu}
\affil{Department of Computer Science\\ 
            Dartmouth}
\date{\today}
\begin{document}
\maketitle 
\thispagestyle{empty}
\begin{abstract}
We consider a scenario in which a database stores sensitive data of users and an analyst wants to estimate statistics of the data. The users may suffer a cost when their data are used in which case they should be compensated. The analyst wishes to get an accurate estimate, while the users want to maximize their utility. We want to design a mechanism that can estimate statistics accurately without compromising users' privacy.

Since users' costs and sensitive data may be correlated, it is important to protect the  privacy of both data \emph{and} cost. We model this correlation by assuming that a user's unknown sensitive data determines a distribution from a set of publicly known distributions and a user's cost is drawn from that distribution. We propose a stronger model of privacy preserving mechanism where users are compensated whenever they reveal information about their data to the mechanism. In this model, we design a Bayesian incentive compatible and privacy preserving mechanism that guarantees accuracy and protects the privacy of both cost and data. 

%We show that the expected payment of this mechanism is 2-approximate to the optimal truthful, ex-post individually rational, and envy-free mechanism that minimizes the expected payment. We also show that the expected payment of this me chanism is 2-approximate to the optimal truthful mechanism that minimizes the expected payment when the distribution satisfies a modified regular condition. For general distributions, we show that the expected payment of this mechanism is $2r$-approximate to the optimal truthful mechanism that minimizes the expected payment, where $r$ is a number depending on the distribution. 
\end{abstract}
\clearpage 
\setcounter{page}{1}
\input{SellingPrivacy_Intro.tex}
\input{SellingPrivacy_Model.tex}
\input{SellingPrivacy_Mechanism.tex}
\input{SellingPrivacy_Opt.tex}

\bibliography{SellingPrivacy}{}
\bibliographystyle{plain}
\input{SellingPrivacy_App.tex}
\end{document}

%% file: SellingPrivacy_Intro.tex
\section{Introduction}
Using the Internet, it is fairly easy to collect sensitive personal data. Online service providers implicitly compensate users who provide their personal data, by offering improved services based on their data. However, this implicit exchange may not be fair to the individual, since different people may have different costs --- a loss in expected utility over future events --- for use of their data. Moreover, companies rarely give well-defined guarantees concerning data privacy and compensation. 
When the compensation is less than the individual's perceived cost, 
the individual may choose not to participate. Here, we explore mechanisms to 
fairly compensate individuals for use of their personal data.
% without compromising their privacy. 

In order to motivate users to participate in a mechanism, the payment 
to a user should be at least the cost to the user. Thus, the mechanism 
should learn information about users' costs. Ghosh and Roth~\cite{Roth11} 
initiate a study of this problem. Their mechanism asks users to report 
their costs for the use of their data to estimate statistics, and then
selects some of the users (based on their stated costs) to determine 
the statistics, and pays these users accordingly.
% and payments based on users' data and reported costs. 
% While their mechanism protects the privacy of the 
% data and also guarantees the accuracy, it does not guarantee the privacy 
% of individual payments. 
This mechanism is problematic when costs and personal data are correlated,
since users may be reluctant to reveal their costs if they are not
guaranteed adequate compensation up front.  
For example, suppose that a database indicates whether a 
vehicle has been damaged. When the database can be publicly accessed, the 
owner of a damaged car cannot sell the car for the same price as the price 
of an undamaged car. Thus, his cost for revealing data is higher than the 
owner of an undamaged car. Revealing information about the costs may also 
reveal information about whether the car is damaged. Thus, it is important 
to also guarantee privacy of individual payments.

We study this problem where costs are correlated with data. 
We model this correlation by assuming that a user's unknown data 
determines a distribution from a set of accurate and publicly known 
distributions and the user's cost is drawn from that distribution.
We propose a model of a privacy preserving mechanism where users are 
compensated whenever they reveal any information about their data to the 
mechanism, whether directly, or indirectly by revealing their costs. 
In this model, we design a Bayesian incentive compatible and individually rational mechanism, 
which produces accurate statistics and protects the privacy of data and costs.

\vspace{0.5em}\noindent{\bf Problem Setting.} 
There are $n$ users, which we call players, denoted by $[n]$. 
Each player has sensitive data $D_i \in [h]$, stored in a database 
$D \in [h]^n$.  Initially $D_i$ is the private information of player $i$.  
However, since $D_i$ is also in the database, it's value may be verified
with player $i$'s permission.  In addition, player $i$ has a value
for his loss of privacy of his data.  This value $v_i$ is private
to player $i$, but it is correlated with $D_i$.  This correlation
is modeled as follows: 
If $D_i = t \in [h]$ then $v_i \sim F_t$, where $F_t$ is a distribution of
privacy costs for players of type $t$ that is known to all players
and the mechanism.
$F_t$ correctly represents the distributions of costs
of type $t$ players.  

A query is a function $Q : [h]^n \rightarrow \mathbb{R}$, mapping a database to a response.
An example of a query is ``what is the number of people $i$ in the database $D$ with $D_i = j$?".
A data analyst wants $Q(D)$. Since the data are sensitive, 
the data analyst accesses the database through a privacy preserving 
algorithm $A$. Therefore, the data analyst does not receive $Q(D)$ 
but receives an estimate $A(D)$. To ensure the estimate is accurate, 
the error $|Q(D) - A(D)|$ should be small with high probability. 

Differential privacy, introduced in~\cite{Dwork06}, is an accepted way 
to measure privacy and privacy preserving algorithms. Two databases $D$ and 
$D'$ are \emph{adjacent} if they differ in only one entry. An algorithm $A$ 
satisfies \emph{$\epsilon$-differential privacy}, where $\epsilon > 0$, 
if for any pair of adjacent database $D$ and $D'$ and any set 
$I \subseteq \mathbb{R}$, $\Pr[A(D) \in I] \leq e^{\epsilon}\Pr[A(D') \in I]$. 
When $\epsilon = 0$, it implies that the algorithm does not depend on $D$. 
If the error $|Q(D) - A(D)|$ is small with high probability, then the 
algorithm should have large $\epsilon$. Thus, privacy guarantees come at 
the expense of the accuracy.

Although an $\epsilon$-differentially private algorithm can protect
sensitive data, if a player allows his data to be used, he may incur a 
cost. We model this cost as linear in the privacy loss $\epsilon$
and his expected cost $v_i$.\footnote{We can view 
this cost as due to the change in his utility from future events 
that depend on the answer he gives to the analyst. 
This cost is approximately linear in $\epsilon$ and his expected utility, 
denoted by $v_i$.  Let $g(A(D))$ be the distribution of future events that depends on $A(D)$. Let $w_i$ be the player $i$'s utility for future events. Since $A$ is $\epsilon$-differentially private, $g \circ A$ is also $\epsilon$-differentially private. Thus, for random variables $y \sim g(A(D))$ and $y' \sim g(A(D'))$ and event $b$, $\Pr[y = b] \leq e^{\epsilon} \Pr[y' = b]$. Therefore, we have $E_{y \sim g(A(D))}[w_i(y)] - E_{y \sim g(A(D'))}[w_i(y)]$ is approximately $\epsilon E_{y \sim g(A(D'))}[w_i(y)]$ or $-\epsilon E_{y \sim g(A(D'))}[w_i(y)]$, when $\epsilon$ is small.} 
Thus, for player $i$ to agree to the use of his data,
his expected payment should be at least $\epsilon v_i$.

% Since $v_i$ is known only to player $i$, in order to determine the payments to the players, the mechanism needs to learn information about $v_i$. But $v_i$ and $D_i$ may be correlated, since the cost is determined by the player's utility over future events. We model the correlation between the costs and the data in 
% the following way: each player has a type in $[h]$. The type of a player determines his data entry and the distribution from which his cost is drawn. That is, player $i$ with type $j$ has data entry $D_i = j$ and cost $v_i$ drawn from a public distribution with cumulative distribution function $F_j$. The assumption that players' costs are drawn from publicly known distributions is standard in the Bayesian optimal auction literature. We also assume that the distributions correctly model the correlation between costs and types. The types of the players are unknown to the mechanism, but if the players agree to their data being used in computing the estimate, the mechanism can verify their true types. 

A mechanism specifies a set of actions that players can take. 
The players take actions based on their data and private costs. 
Thus, the input of the mechanism is a database and a vector of actions. 
The outputs are an estimate $\hat{s}$ and a payment vector 
$p = (p_1, \dots, p_n)$. Since player $i$ has a linear cost $\epsilon v_i$, 
the utility of player $i$ is $p_i - \epsilon v_i$ if $D_i$ is used in the 
mechanism, otherwise the utility is $p_i$. 
We assume that all players are rational and want to maximize their utilities. 
A mechanism is a \emph{direct} mechanism if the action set equals the set of all real numbers.
That is, a direct mechanism asks players to report their costs.
A direct mechanism is truthful if every player reports his true cost in order to maximize his utility.
Truth telling is a concept defined for direct mechanisms. 
In this paper, we propose an indirect mechanism. 
Thus, we want to extend the notion of truthfulness to indirect mechanisms.
In our mechanism, there is a straightforward mapping, described in Section \ref{sec:mechanism}, from player's type set to player's action set.
We say that a player \emph{decides truthfully} if he picks the strategy corresponding to his type under this mapping.

In our paper, we will assume that the query/goal of the analyst is to
estimate $n_j = |\{i : D_i = j\}|$.
Without loss of generality, we assume throughout the paper that the 
data analyst wants to estimate $n_1$. 
We seek to design a mechanism with the following properties.
\begin{enumerate}
\item \textbf{Accuracy:} A mechanism $M$ is \emph{$k$-accurate}, if for any 
database $D$, $\Pr[|\hat{s} - n_1| \geq k ] \leq \frac{1}{3}$, when every 
player decides truthfully. Note that the accuracy guarantee is independent 
of the size of the database --- the number $k$ is fixed no matter how large 
the database is, or the sampled set is.
\item \textbf{Differential Privacy:} The estimate and payments satisfy 
$\epsilon$-differential privacy.
\item \textbf{Truthfulness:} A mechanism is \emph{dominant strategy truthful} if, for every player, deciding truthfully maximizes his utility. A mechanism is \emph{Bayesian incentive compatible} (BIC) if, for every player, assuming that other players' costs are drawn from $F$ according to their data and decide truthfully, deciding truthfully maximizes his utility.
\item \textbf{Individual Rationality:} If a player's utility is non-negative, 
then he should be willing to participate. A mechanism is 
\emph{ex-post individually rational} (EPIR) if the utility is non-negative 
for every player when he decides truthfully. A mechanism is 
\emph{ex-interim individually rational} (EIIR) if the expected utility 
is non-negative for every player when he decides truthfully, where the randomness comes from the
mechanism and the costs of other players.
\item \textbf{Payment Minimization:} The summation of payments should be 
as little as possible.
\end{enumerate}

To get permission to use a player's data, the mechanism must compensate
the player by at least his perceived loss of privacy.  But since costs are
correlated with data, players may be reluctant to reveal their true costs,
unless they will be compensated for this.  To avoid this seeming 
chicken-and-egg problem, the mechanism designer cannot resort to the revelation
principle, which states that any mechanism can be realized
as a direct and truthful mechanism.
In fact, \cite{Roth11} prove that if costs and data can be arbitrarily 
correlated and player's cost of privacy can be unbounded, then for any 
$k < n/2$, no $k$-accurate, direct, dominant strategy truthful, EPIR, privacy preserving 
mechanism exists. On the other hand, 
we give a mechanism that provides $k$-accuracy for 
any input value $k$ when costs are correlated with data, and there is no bound on players' cost of privacy.
We get around the lower bound of \cite{Roth11} by using an 
\emph{indirect} mechanism, and modeling the correlation of values
and data via publically known (and allowably unbounded) distributions.
%Distributional assumptions are common in the mechanism design literature~\cite{Hart09},
%and are useful since they often lead to natural mechanisms with provably 
%good and predictable outcomes. 
% for problems that would otherwise be intractable.
%XXX (needs reference, and check correctness of claims)

\vspace{0.5em}\noindent{\bf Privacy Issues when Costs are Correlated with Data.}
The objective of a privacy preserving mechanism is that the increase in knowledge about a player's data due to output of the mechanism is small.
Previous work on privacy in statistical 
databases assumes that the mechanism is associated with the database, such that the mechanism can access the whole database without compromising a player's privacy. However, if the mechanism is separated from the database, then a player might not trust the mechanism and might not want to reveal private information to the mechanism.

In our problem, in order to estimate $n_1$, the mechanism should learn 
information about players' data. Suppose that the mechanism has a prior 
belief $G$ about the data in $D$. That is, the mechanism believes that the 
probability of $D_i = j$ is $\Pr_G[D_i = j]$ according to the prior belief. 
The mechanism \emph{learns} about $D_i$ if the mechanism believes that 
$\Pr[D_i = j] \neq \Pr_{G}[D_i = j]$ after running the mechanism, for some 
$j$. There are two possible ways to learn about players' data. The first 
way is to read $D_i$ explicitly. The second way is to read players' 
actions and deduce something about their $D_i$. For example, 
if the mechanism is direct and truthful, then the players report $v_i$ 
truthfully. Suppose that the prior belief is that every player's data are 
drawn from a uniform distribution. That is, $\Pr_{G}[D_i = j]$ is the same 
for all $i$ and $j$. If $F_j(v_i) < F_{j'}(v_i)$ for some $j$ and $j'$, and 
player $i$ truthfully reports $v_i$, then the mechanism's posterior belief is that $\Pr[D_i = j] < \Pr[D_i = j']$, which is different from the prior belief. Learning anything about a player's data may compromise a player's privacy and should be compensated. Thus, there are two kinds of cost to a player that should be compensated, one is for using the player's data and one is for learning about the player's data.

For the latter cost, we propose the concept of perfect data privacy, which is inspired by the concept of perfect objective privacy introduced in~\cite{FJM10}. A mechanism satisfies \emph{perfect data privacy} if whenever the mechanism's posterior belief about a player's data differs from its prior belief, the mechanism pays the player. Under perfect data privacy, mechanisms can learn about a player's cost, as long as that knowledge does not reveal anything about his data. 
However, for a perfectly data private mechanism, if the mechanism learns about a player's data, then the mechanism always compensates the player, even when the mechanism does not not use the player's data to compute the estimate.

\vspace{0.5em}\noindent{\bf Our Main Contribution.}
We give a mechanism that is BIC, EIIR, $O(\epsilon^{-1})$-accurate, perfectly data private, and $\epsilon$-differentially private. 
To achieve our privacy guarantees, we propose a posted-price-like mechanism, described in Section \ref{sec:mechanism}. Given the set of types of players and the distributions of costs, the mechanism writes a contract that offers a different expected payment for each type. Each player is offered this contract. If a player accepts the contract, then his payment is determined by his verifiable type and the payment for his type in the contract. The player's action is either to accept the contract or reject the contract. A player \emph{decides truthfully} if a player with type $j$ accepts the contract when $\epsilon v_i \leq r_j$, where $r_j$ is the payment for type $j$ in the contract. We prove that this posted-price-like mechanism is BIC, EIIR, $O(\epsilon^{-1})$-accurate, perfectly data private, and $\epsilon$-differentially private.

% These next two paragraphs need to be changed
We seek a mechanism with a small payment.  
In Section~\ref{sec:opt}, we define a benchmark for the expected payment of a mechanism and compare the expected payment of our mechanism to this benchmark in two different settings.
When costs are non-negative, we show that our mechanism is close to the benchmark.

We also prove a lower bound on the accuracy that a direct and data 
private mechanism can achieve in Section~\ref{sec:bmd}. 
\subsection{Related Work} 

\noindent{\bf Selling Privacy.}
Our paper is closely related to the privacy preserving mechanisms studied in~\cite{Roth11}. In~\cite{Roth11}, they extend the definition of $\epsilon$-differentially private algorithms to $\epsilon$-differentially private mechanisms. Under their definition of an $\epsilon$-differentially private mechanism, the randomness only comes from the mechanism. In our model, since we want to protect the privacy of the costs, which are drawn from distributions, our definition of an $\epsilon$-differentially private mechanism relies both on the distributions of the costs and the randomness of the mechanism.

\vspace{0.5em}\noindent{\bf Differential Privacy.}
A comprehensive survey of differential privacy appears in~\cite{Dwork08}. 
Most of the previous results are based on random perturbations of the output, 
% For example, Dwork et al.~\cite{Dwork06} use Laplacian distribution as a 
% source of randomness, while Ghosh et al.~\cite{Ghosh09} use two-sided 
% geometric distribution as a source of randomness. However, most of the 
% previous results 
and assume that the mechanism has the ability to access the 
whole database. If the mechanism cannot access the whole database, Chaudhuri et al.~\cite{ChaudhuriM06} and Klonowski et al.~\cite{Klonowski10} show that random sampling is enough to ensure differential privacy with high probability. That is, it is not necessary to add more noise to the output. 
% There is a growing list on issues in differential privacy. For example, \cite{Chao11} and \cite{Roth10} discuss the issue of multiple queries.

\vspace{0.5em}\noindent{\bf Differential Privacy and Mechanism Design.}
McSherry et al.~\cite{MT07} use a privacy preserving algorithm as a tool 
to design an approximately dominant strategy truthful mechanism. 
% If a mechanism is $\epsilon$-differentially private, then any player can affect the outcome of the mechanism by a factor at most $\epsilon$. This property guarantees that the mechanism is approximately truthful, that is, players can increase their utilities by at most a factor $\epsilon$ by lying. However, Nissim et al.~\cite{NST10} point out that besides truth-telling, \emph{any} strategy is also an approximately dominant strategy in a differential private mechanism. Thus, reporting truthfully does not have advantage over other strategies.
% In this paper, we discuss a different problem. We do not use differential privacy to design an approximately truthful mechanism. Instead, we use differential privacy as a quantification of privacy and treat sensitive data as a commodity that we want to buy.
Instead, we focus on treating senstive data as a commodity that can be
sold.

\vspace{0.5em}\noindent{\bf Privacy Concerns in Mechanisms.}
Traditional mechanism design theory focuses on drawing private information from players in order to compute a result. However, if players have privacy concerns, they may not want to reveal their information. Feigenbaum et al.~\cite{FJM10} study how to quantify the information leakage to the mechanism based on communication complexity. 
%They introduced the concept of perfect objective privacy. A mechanism satisfies perfect objective privacy if the mechanism learns about players' private information no more than necessary to compute the outcome. This concept differs from differential privacy, which guarantees an \emph{outsider} cannot learn much about players' private information from the outcome. They define \emph{privacy-approximation ratio} (PAR), as the ratio between the number of bits the mechanism reads and the necessary number of bits to compute the outcome. They analyze the PAR of the second price auction and several other mechanisms. 
%The concept of perfect objective privacy is similar to perfect data privacy. The major difference is that in the perfect objective privacy, the mechanism has no prior knowledge about players' private information and they quantify the information leakage based on communication complexity. In perfect data privacy, the mechanism has a prior knowledge and we measure the information leakage based on that prior knowledge.

Xiao~\cite{Xiao11} quantifies the information leakage in a mechanism based on information theory. 
In his model, the outcome of a privacy preserving mechanism not only motivates the players to participate but also protects the private information of players. 
%Thus, besides the utility for the outcome of the mechanism, players also have costs for the information leakage from the mechanism. 
In independent work, Nissam et al.~\cite{Nissim11} and Chen et al.~\cite{Chen11}
consider privacy issues in mechanism design in the context of
elections and discrete facility location.
%However, their models are different from ours. First, the goal of their works is to construct a truthful mechanism which guarantees the privacy of players. Our mechanism has one more objective, that is, minimizing the payment. Second, we assume that the players' privacy costs are drawn from distributions and the domains of distributions can be unbounded. However, both ~\cite{Nissim11} and ~\cite{Chen11} require that the costs are from a known bounded range.

\vspace{0.5em}\noindent{\bf Posted-Price Mechanisms.}
In a posted-price mechanism, player $i$ is offered a price $r_i$. If player $i$ accepts that price, then $i$ pays $r_i$ to get the allocation. Goldberg et al.~\cite{GH05} show that the posted-price mechanism is collusion resistant. Moreover, the players do not need to know or report their private values precisely. They only decide to accept or reject the price. Chawla et al.~\cite{CH10} point out that this could be useful in reducing the private information revealed to the mechanism.

\vspace{0.5em}\noindent{\bf Revenue Maximization in Bayesian Mechanism Design.}
In a classic paper, Myerson~\cite{My81} characterizes the optimal BIC 
selling mechanism to maximization the expected revenue. In procurement 
mechanisms, each player is a supplier and each player's production cost is 
private information.  The auctioneer is the buyer and wants to minimize the 
expected payment.  In the computer science literature, an early
paper in this area characterizes the minimum-cost dominant strategy truthful auction to buy an
s-t path in a graph~\cite{Elkind04}.  
Since then, there has been considerable interest
in both \emph{frugal} mechanism design (buying a feasible set at low cost), and
budget-constrained mechanism design (buying as good a set as possible subject
to a budget).  Our work can be seen as a generalization of these questions
to the setting of bidders who are reluctant to reveal their costs,
and the feasibility of a set depends on the private costs (via
the correlation with data).

%% file: SellingPrivacy_Model.tex
\section{Model and Lower Bound} \label{sec:bmd}
\subsection{Model}
There is a database $D \in [h]^n$ and $n$ players, where each player has data $D_i$. Player $i$ with $D_i = j$ has a private cost $v_i$ drawn from a distribution with cumulative distribution function $F_j$. Note that this definition is different from the traditional definition of a Bayesian setting. In the traditional definition, the distribution of $v_i$ is known to every player and the mechanism. In our definition, the mechanism and players know that each player's $v_i$ is drawn from one of a set of distributions, but the particular distribution depends on the individual player's data, which is unknown to everyone but that player.

The goal of our mechanism is to estimate $n_1$ based on $D$ and determine the payment $p_i$ for every player $i$. 
A mechanism first specifies the set of possible actions $Y$ that players can take. 
Then, based on players' actions and the database, the mechanism determines the estimate and payment. Formally, a mechanism is a function $M : Y^n  \times [h]^n \rightarrow \mathbb{R} \times \mathbb{R}^n$.
%That is, a mechanism takes a database and players' actions as input and output an estimate and payment for every player.
The mechanism has an a priori belief $G$ about the data in $D$. 
That is, the mechanism believes that the probability of $D_i = j$ is $\Pr_G[D_i = j]$. 
Recall that the mechanism learns about $D_i$ if, after running the mechanism, the mechanism believes that $\Pr[D_i = j] \neq \Pr_{G}[D_i = j]$ for some $j$. 
%The mechanism learns about players' data by accessing player's data $D_i$ and accessing player's action $y_i$. 
We use a vector $x \in \{0, 1\}^n$ to indicate whether the mechanism learns something about each player's data. 
If the mechanism learns about $D_i$, then $x_i = 1$.
A mechanism is \emph{perfectly data private} if, when $x_i = 1$, player $i$'s expected payment from the mechanism is at least $\epsilon v_i$.
We focus on randomized mechanisms in this paper, that is, $x_i$ and payment $p_i$ are random variables.

Next, we define the utility for a player.
If $x_i = 1$, there is a cost $\epsilon v_i$ to player $i$, since something about $D_i$ is learned.
For $y \in Y^n$ representing all players' actions, the \emph{utility} for player $i$ is $u_i(y, v_i) = p_i - \epsilon x_i v_i$, where $(\hat{s}, p) = M(y, D)$. 
In this paper, we assume that players are rational, so players want to maximize their expected utilities. 
The \emph{strategy} of player $i$ is a function $q_i : \mathbb{R} \times [h] \rightarrow Y$ mapping from $v_i$ and $D_i$ to an action. 
Since players want to maximize their expected utilities, they will take the action that is not worse than any other action. 

Finally, we introduce the solution concept.
A profile of strategies $q_1, \dots, q_n$ is a \emph{Bayesian-Nash equilibrium} if for all $i$, $v_i$, and $y_i' \in Y$, $E[u_i(q(v_i, v_{-i}, D), v_i)] \geq E[u_i((y_i', q_{-i}(v_{-i}, D_{-i})), v_i)]$, where the randomness is from the mechanism and the randomness of $v_{-i}$.
A direct mechanism is \emph{Bayesian incentive compatible} (BIC) if $q_i(v_i, D_i) = v_i$ is a Bayesian-Nash equilibrium for every player $i$.

\subsection{Lower Bound}
In order to ensure that players have incentive to participate the mechanism, we wish that the mechanism is individually rational. However, we can show that for any direct, BIC, and EIIR mechanism, there is a lower bound of accuracy. Since the condition of EIIR is weaker than EPIR, the lower bound for EIIR also implies a lower bound for EPIR mechanisms.
\begin{lemma} \label{lem:lower}
If the functions $F_i$ are arbitrary functions with unbounded range, then for any $k < n/2$, no $k$-accurate, direct, BIC, EIIR, and perfectly data private mechanism exists.
\end{lemma}
\begin{proof}
Suppose that $M$ is a BIC, EIIR, perfectly data private, and $k$-accurate mechanism. First, we show that $M$ must access at least one player's cost or data. Assume that $M$ does not access any cost or data. Thus, $M$ randomly output an estimate $\hat{s}$, which is independent of costs and data. For a database $D^1$ with all entries equal to one, since $M$ is $k$-accurate, $\Pr[\hat{s} \in [n, n-k] ] \geq \frac{2}{3}$. Similarly, if a database $D^0$ has no entries equal to one, then $\Pr[\hat{s} \in [0, k]] \geq \frac{2}{3}$. Because $k < n/2$, $[n, n-k]$ and $[0, k]$ do not overlap. But the summation of these two probabilities is greater than one, which is impossible. Hence, $M$ must access at least one player's cost or data.

Suppose that $D_i \in \{ 1, 2 \}$ and $F_1(v) \neq F_2(v)$ for all $v$. For any $\hat{v}$, if $M$ access $v_i = \hat{v}$, then the mechanism must pay player $i$, since $F_1(\hat{v}) = \Pr[v_i=\hat{v}| D_i = j] \neq \Pr[v_i=\hat{v}| D_i = j'] = F_{2}(\hat{v})$ and $M$ is perfectly data private. Let $x_i$ be the indicator random variable representing whether player $i$'s cost is accessed. Let $p_j$ be the random variable representing player $i$'s payment. Since $M$ is BIC, we suppose that players other than $i$ report truthfully. Since the mechanism decides to access $v_i$ based on $v_{-i}$, $\Pr[x_i = 1]$ is independent of $v_i$. Because $M$ must access at least one player's cost, we can find a player $i$, such that $\Pr[x_i = 1] > 0$. For a fixed $v_i$, the expected utility of $i$ is $E[p_i] - \epsilon v_i E[x_i]$. Since the range of $F$ is unbounded, we can find another $v_i' > \frac{E[p_i]}{\epsilon E[x_i]}$. Since $M$ is EIIR, we have $E[p_i'] \geq \epsilon v_i' E[x_i]$. Thus, for player $i$ with cost $v_i$, if $i$ overbids $v_i'$, the utility is $E[p_i'] - \epsilon v_i E[x_i] \geq \epsilon v_i' E[x_i] - \epsilon v_i E[x_i] > E[p_i] - \epsilon v_i E[x_i]$. Thus, player $i$ can increase expected utility by overbidding. Hence, $M$ is not BIC.
\end{proof}

Our mechanism, which is explained in the next section, is an indirect mechanism since it does not ask for players' costs. The revelation principle, which states that if there exists an indirect mechanism implementing a function in Bayesian-Nash equilibrium, then there also exists a direct BIC mechanism implementing the same function, is irrelevant under the desire for perfect data privacy. It is easy to construct a direct mechanism from our indirect mechanism. However, this direct mechanism accesses all players' data without compensating all players. Thus, this direct mechanism is not perfectly data private.

\vspace{0.5em}\noindent{\bf $\epsilon$-Differential Privacy.} The traditional definition of $\epsilon$-differential privacy compares the outcomes of the algorithm applied to adjacent databases. However, with a mechanism that offers payments, the mechanism may use both the database and the replies to the mechanism to compute an estimate and payments. Since replies depend on the individuals' costs, we compare the outcomes of the mechanism applied to two cost-data pairs $(v, D)$ and $(v', D')$. A cost vector $v = (v_1, \dots, v_n)$ is \emph{drawn according to a database $D$}, if $v_i$ is drawn from $F_j$, where $D_i = j$. Two cost-data pairs $(v, D)$ and $(v', D')$ are \emph{adjacent}, if $D$ and $D'$ differ only in the $i$-th entry and $v$ and $v'$ are independently drawn according to database $D$ and $D'$. A BIC mechanism is $\epsilon$-differentially private if, for any pair of adjacent cost-data pairs, the estimate and payments satisfy $\epsilon$-differential privacy.

\vspace{0.5em}\noindent{\bf Bayesian Assumptions.}
Our definition of $\epsilon$-differential privacy is based on the common belief $F$.
That is, the player decides his strategy assuming that other players' costs are drawn from $F$ and all players believe this assumption.
If a player allows his data to be used, then he may incur a expected cost $\epsilon v_i$.
The expected cost to the player depends on $\epsilon$ and thus also depends on the common belief $F$.
Having a common belief is a traditional assumption in the Bayesian setting.
Moreover, most BIC mechanisms become meaningless when the common belief is not true.
Thus, we assume that the common belief $F$ is correct.

%% file: SellingPrivacy_Mechanism.tex
\section{Mechanism} \label{sec:mechanism}
In this section, we give a perfectly data private, BIC, EIIR, $\epsilon$-differentially private, and $O(\epsilon^{-1})$-accurate mechanism. 
Every player $i$ has data $D_i \in [h]$.
To start, we assume that $F_j$ is continuous for $j \in \{1, 2\}$. 

The mechanism designs and offers contracts to players. 
The contract guarantees an expected payment to each player who accepts the contract. 
The players decide to accept or reject the contract. 
Thus, the possible actions for players are ``accept" or ``reject". 
The mechanism uses the data of players who accept the contract to estimate $n_1$. 
The estimate is \emph{unbiased} if the expected value of the estimate is $n_1$. 
To obtain an unbiased estimate, the set of players who accept the contract should be unbiased, that is, the probability of a player accepting the contract should be equal for all players. 
Moreover, since the mechanism pays players, the costs of players in the accepting set should be bounded. 

The mechanism first finds $\alpha_j$ for $j \in [h]$, such that $F_j(\alpha_j) = c$, where $c$ will be determined later. 
Then, each player $i$ is given a contract : ``If $D_i = j$, your expected payment will be $\epsilon \alpha_j$." 
A player $i$ with $D_i = j$ decides truthfully if, when $v_i \leq \alpha_j$, player $i$ accepts the contract and rejects otherwise.
Let $W$ be the set of players who accept the contract. 
If all players decide truthfully, the cost to each player in $W$ is bounded by $\max_j \alpha_j$. 
Since for player $i$ with $D_i = j$, $\Pr[v_i \leq \alpha_j] = c$, every player accepts the contract with probability $c$. 
Thus, $W$ is an unbiased and cost-bounded sample set.

Since the probability that a player accepts the contract is $c$, the value $m := |\{i \in W : D_i = 1\}|$ is a random variable $\textrm{bin}(n_1, c)$ from a binomial distribution\footnote{A \emph{binomial distribution} with parameter $n$ and $p$ is denoted by $\textrm{Bin}(n, p)$. 
The probability density function of $\textrm{Bin}(n, p)$ is $f(k; n, p) = {n \choose k} p^k (1-p)^{n-k}$. Let $\textrm{bin}(n, p)$ denote a random variable drawn from $\textrm{Bin}(n, p)$. 
The expected value of $\textrm{bin}(n, p)$ is $np$ and variance is $np(1-p)$.} $\textrm{Bin}(n_1, c)$. 
Since the expected value of $m$ is $cn_1$, $\frac{m}{c}$ is an unbiased estimate of $n_1$. We say $\frac{m}{c}$ is a \emph{na\"{\i}ve} estimate of $n_1$.

We explain how to produce an estimate that satisfies $\epsilon$-differential privacy.
Although the naive estimate is an unbiased estimate of $n_1$, it does not satisfy differential privacy.
Consider an adjacent pairs of cost-data pairs $(v, D)$ and $(v', D')$, where $D$ and $D'$ differ in the $i$-th entry.
Let $n_1$ be the number of player $i$ with $D_i = 1$ and $n_1'$ be the number of players $i$ with $D_i' = 1$.
The naive estimate does not satisfy differential privacy, since if $D_i = 1$ and $v_i \leq \epsilon \alpha_1$, then an outsider can infer $D_i$ easily by comparing the naive estimates of $n_1$ and $n_1'$.
Thus, we should introduce a random noise to the naive estimate to satisfy differential privacy.

The mechanism uses the Laplacian distribution as a source of the random noise. 
The Laplacian noise is commonly used to obtain differential privacy.
A \emph{Laplacian distribution} with mean 0 and parameter $b > 0$ is denoted by $\textrm{Lap}(b)$. The probability density function of $\textrm{Lap}(b)$ is 
\[f(x) = \frac{1}{2b} \exp\Paren{-\frac{|x|}{b}}.\]
Let $\textrm{lap}(b)$ denote a random variable drawn from $\textrm{Lap}(b)$.

In order to make estimate satisfy differential privacy, the mechanism adds random noise $\textrm{lap}(\frac{1}{\epsilon})$ to the naive estimate. 
Since the mean of the Laplacian noise is zero, $s = \frac{1}{c}(m + \textrm{lap}(\frac{1}{\epsilon}))$ is an unbiased estimate of $n_1$. 
However, $s$ might be larger than $n$ or be negative, both of which are meaningless. 
We \emph{truncate} $s$ to get $\hat{s}$, that is when $s > n$, the mechanism outputs $n$ and when $s < 0$, the mechanism outputs $0$.

We also use the Laplacian noises to produce payments that satisfy $\epsilon$-differential privacy. 
By the construction of the contract, for any player $i$ with $D_i = j$ who accepts the contract, the mechanism pays player $i$ for $\epsilon \alpha_j$ in expectation.
If the mechanism pays player $i$ for $\epsilon \alpha_j$ deterministically, then an outsider can infer player $i$'s data easily.
Thus, we should introduce noise to the payments.
We add noise $\epsilon~\textrm{lap}(\frac{\gamma}{\epsilon})$ to the payment, where $\gamma := |\max_j \alpha_j - \min_j \alpha_j|$. 
Thus, $p_i = \epsilon(\alpha_j + \textrm{lap}(\frac{\gamma}{\epsilon}))$. 
Since the expected value of $\textrm{lap}(\frac{\gamma}{\epsilon})$ is zero, the expected payment of player $i$ is $\epsilon \alpha_j$, which satisfies the guarantee in the contract. 
Moreover, since $\epsilon \alpha_j$ is larger than $\epsilon v_i$, the mechanism is EIIR. The formal description of the mechanism is in Mechanism \ref{mechanism}.

\begin{algorithm}[t]
    \AlCapSty{\AlTitleFnt{Mechanism 1: $\epsilon$-differentially private mechanism}}\\
    \label{mechanism}
    \DontPrintSemicolon
    \LinesNumbered
    \SetKwInOut{Input}{input}
    \SetKwInOut{Output}{output}
    \Input{privacy parameter $\epsilon$; cost distributions $F_j$, $j \in [h]$}
    \Output{estimate $\hat{s}$; payment $p$}
Pick a real number $c \in (0, 1)$\\
Find $\alpha_j$ for all $j \in [h]$, such that $F_j(\alpha_j) = c$.\\
For each player $i$, offer a contract:  \\
\Indp If $D_i = j$, the expected payment will be $\epsilon \alpha_j$. \\
\Indm Let $W = \{ i : \text{$i$ accepts contract} \}$. \\
Let $m = |\{ i \in W : D_i = 1 \}|$. \\
Let $s = \frac{1}{c}(m + \textrm{lap}(\frac{1}{\epsilon}))$.\\
$\hat{s} = s \text{ if s $\in [0, n]$}, \text{ }   0  \text{ if $s < 0$}, \text{ } n \text{ if $s > n$}$ \\
$p_i = \begin{cases}
       0    & \text{if $i \notin W$}\\
       \epsilon (\alpha_j + \textrm{lap}(\frac{\gamma}{\epsilon}))\text{, where $\gamma :=  |\max_j \alpha_j - \min_j \alpha_j|$} & \text{if $i \in W$ and $D_i = j$} \\
       \end{cases}$ \\
\Return{$(\hat{s}, p)$}
\end{algorithm}
\begin{lemma}\label{lem:safe}
Mechanism 1 is perfectly data private.
\end{lemma}
\begin{proof}
Let $y_i$ be player $i$'s reply to the contract. By construction of the contract, if $i$ decides truthfully, then $\Pr[y_i = ``accept" \mid D_i = j] = c$ for all $j \in [h]$. That is, the probability of accepting the contract and $D_i$ are independent. Thus, for any $i$, the mechanism cannot learn about $D_i$ by reading $y_i$. Moreover, the mechanism only reads $D_i$, where $i \in W$. Since player $i \in W$ with $D_i = j$ is paid $\epsilon \alpha_j$ in expectation and $v_i \leq \alpha_j$, the mechanism satisfies the requirement.
\end{proof}
\begin{lemma} \label{lem:truthful}
Mechanism 1 is BIC and EIIR. 
\end{lemma}
\begin{proof}
(BIC) The payments for players who is not in $W$ are always 0. For player $i$, there are two cases.\\
\textbf{Case 1:} $D_i = j$ and $v_i \leq \alpha_j$. Accepting the contract will get expected payment $\epsilon(\alpha_j - v_i) \geq 0$.\\
\textbf{Case 2:} $D_i = j$ and $v_i > \alpha_j$. Accepting the contract will get expected payment $\epsilon(\alpha_j - v_i) < 0$.

(EIIR) Suppose that every player decides truthfully. Then only players with $v_i \leq \alpha_j$ and $D_i = j$ for some $j$ are in $W$. Since the expected payment for $i$ with $D_i = j$ is $\epsilon \alpha_j$, the expected utility of the player is non-negative.
\end{proof}

%We will abuse the notion of $\epsilon$-differential privacy by extending this term to two random variables. 
Two random variables $x_1$ and $x_2$ are $\epsilon$-\emph{mutually bounded}, if $\forall I \subseteq \mathbb{R}$, $\Pr[x_1 \in I] \leq e^{\epsilon} \Pr[x_2 \in I]$ and $\Pr[x_2 \in I] \leq e^{\epsilon} \Pr[x_1 \in I]$. 

\begin{lemma}[Fact 2 in~\cite{Roth11}]\label{prop:dp}
If $x_1$ and $x_2$ are $\epsilon$-mutually bounded and $f$ is a function, then $f(x_1)$ and $f(x_2)$ are also $\epsilon$-mutually bounded. \qed
\end{lemma}

\begin{lemma}[\cite{Dwork06}]\label{prop:dp2}
Let $x_1$ and $x_2$ be two random variables. If $|x_1 - x_2| \leq k$, then $x_1 + \textrm{lap}(\frac{k}{\epsilon})$ and $x_2 + \textrm{lap}(\frac{k}{\epsilon})$ are $\epsilon$-mutually bounded. \qed
\end{lemma}

The next two lemmas address the $\epsilon$-differential privacy of the payment and the estimate. Let $(v, D)$ and $(v', D')$ be adjacent cost-data pairs. Let $(\hat{s}, p)$ and $(\hat{s}', p')$ be the results for $(v, D)$ and $(v', D')$ respectively.

\begin{lemma} \label{lem:is-privacy}
For any $I \subseteq \mathbb{R}$, $\Pr[ \hat{s} \in I ] \leq e^{\epsilon}\Pr[\hat{s}' \in I]$.
\end{lemma}
\begin{proof}
Without loss of generality, we assume that $1 = D_i$ and $D_i' \neq 1$. First, $\Pr[\hat{s} \in I] = \int_{v_{-i} \in \mathbb{R}^{n-1}} \Pr[\hat{s} \in I \mid v_{-i}] \Pr[v_{-i}] dv_{-i}$. Similarly, $\Pr[\hat{s}' \in I] = \int_{v_{-i} \in \mathbb{R}^{n-1}} \Pr[\hat{s}' \in I \mid v_{-i}] \Pr[v_{-i}] dv_{-i}$. Let $\hat{q}_w$ and $\hat{q}'_w$ be two random variables, which are equal to $\hat{s}$ and $\hat{s'}$ when $v_{-i} = w$. If $\hat{q}_w$ and $\hat{q}'_w$ are $\epsilon$-mutually bounded for all $w$, then $\hat{s}$ and $\hat{s}'$ are $\epsilon$-mutually bounded, since then
{
\allowdisplaybreaks
\begin{align*}
\Pr[\hat{s} \in I] &= \int_{w \in \mathbb{R}^{n-1}} \Pr[\hat{s} \in I \mid v_{-i} = w] \Pr[v_{-i} = w] dw \\
& = \int_{w \in \mathbb{R}^{n-1}} \Pr[\hat{q}_w \in I] \Pr[v_{-i} = w] dw \\
& \leq \int_{w \in \mathbb{R}^{n-1}} e^{\epsilon} \Pr[\hat{q}'_w \in I] \Pr[v_{-i} = w] dw \\
& = \int_{w \in \mathbb{R}^{n-1}} e^{\epsilon}\Pr[\hat{s}' \in I \mid v_{-i} = w] \Pr[v_{-i} = w] dw ~=~ e^{\epsilon} \Pr[\hat{s}' \in I].
\end{align*}
}

The case $\Pr[\hat{s}' \in I] \leq e^{\epsilon} \Pr[\hat{s} \in I]$ can be shown by a symmetric argument.

Here, we show that $\hat{q}_w$ and $\hat{q}_w'$ are $\epsilon$-mutually bounded for all $w$. Fix $v_{-i} = w$. Let $W_w$ and  $W_w'$ be the sets of players accepting the contract when applying the algorithm to inputs $(v, D)$ and $(v', D')$ respectively. Let $m_w := |\{ i : D_i = 1, i \in W_w \}|$ and $m'_w := |\{ i : D_i' = 1, i \in W_w' \}|$. When applying the mechanism to inputs $(v, D)$ and $(v', D')$, the mechanism computes $s_w = \frac{1}{c}(m_w + \textrm{lap}(\frac{1}{\epsilon}))$ and $s_w' = \frac{1}{c}(m_w' + \textrm{lap}(\frac{1}{\epsilon}))$ respectively. Then, the mechanism truncates $s_w$ and $s_w'$ to get $\hat{s}_w$ and $\hat{s}_w'$. By Lemma \ref{prop:dp}, since multiplication and truncation are functions, it suffices to show that $m_w + \textrm{lap}(\frac{1}{\epsilon})$ and $m_w' + \textrm{lap}(\frac{1}{\epsilon})$ are $\epsilon$-mutually bounded when $v_{-i} = w$. Since $W \setminus W'$ is either the empty set or $\{i\}$, the difference between $m_w$ and $m_w'$ is at most one. Thus, Lemma \ref{prop:dp2} implies that $m_w + \textrm{lap}(\frac{1}{\epsilon})$ and $m_w' + \textrm{lap}(\frac{1}{\epsilon})$ are $\epsilon$-mutually bounded. Thus, $\hat{q}_w$ and $\hat{q}_w'$ are $\epsilon$-mutually bounded for all $w$, and hence $\hat{s}$ and $\hat{s}'$ are mutually bounded.
\end{proof}

\begin{lemma} \label{lem:ip-privacy}
For all $i \in [n]$ and for all $I \subseteq \mathbb{R}$, $\Pr[ p_i \in I ] \leq e^{\epsilon}\Pr[p_i' \in I]$. 
\end{lemma}
\begin{proof}
Without loss of generality, we assume that $D_i = 1$ and $D_i' \neq 1$. For player $j \neq i$, if $j \notin W$, the payment is zero. If $j \in W$, the payment to $j$ depends only on the data $D_j$ and does not depend on the set of players receiving payments. Thus, $p_j$ does not change and we only need to consider $p_i$. Note that $p_i \neq 0$ only happens if player $i$ is in $W$. If $p_i \neq 0$, then $p_i$ is a random variable $P^1 = \epsilon(\alpha_1 + \text{\textrm{lap}($\frac{\gamma}{\epsilon}$)})$. Thus, for any $I \subseteq \mathbb{R} \setminus \{0\}$, the probability $\Pr[p_i \in I] = c\Pr[P^1 \in I]$, where $c$ is the probability of that a player accepts the contract. The probability $\Pr[p_i = 0] = (1-c) + c\Pr[P^1 = 0]$. Suppose that $D_i' = j'$. Symmetrically, let $P^2 = \epsilon(\alpha_{j'} + \text{\textrm{lap}($\frac{\gamma}{\epsilon}$)})$, for any $I \subseteq \mathbb{R} \setminus \{0\}$, the probability $\Pr[p_i' \in I] = c\Pr[ P^2 \in I]$ and $\Pr[p_i' = 0] = (1-c) + c\Pr[P^2 = 0]$. 

Thus, it suffices to show that $P^1$ and $P^2$ are $\epsilon$-mutually bounded. By Lemma \ref{prop:dp}, since multiplication is a function, it is sufficient to show that $\alpha_1 + \textrm{lap}(\frac{\gamma}{\epsilon})$ and $\alpha_{j'} + \textrm{lap}(\frac{\gamma}{\epsilon})$ are $\epsilon$-mutually bounded. By Lemma \ref{prop:dp2}, since $|\alpha_1 - \alpha_{j'}| \leq \gamma$, $\alpha_1 + \textrm{lap}(\frac{\gamma}{\epsilon})$ and $\alpha_{j'} + \textrm{lap}(\frac{\gamma}{\epsilon})$ are $\epsilon$-mutually bounded.
\end{proof}

\begin{lemma} \label{lem:d-accuracy}
Mechanism 1 is $\sqrt{3(\frac{n_1(1-c)}{c}+\frac{2}{\epsilon^2c^2})}$-accurate. 
\end{lemma}
\begin{proof}
Since the error term $|\hat{s} - n_1|$ is smaller than $|s - n_1|$, we can analyze $|s - n_1|$ to get a bound on the error. Since $E[m] = cn_1$, $E[s] = \frac{1}{c}(E[m] + E[\textrm{lap}(\frac{1}{\epsilon})]) = n_1$ by linearity of expectation.
{
\allowdisplaybreaks
\begin{align*}
|\hat{s} - n_1| ~ \leq ~ |s - n_1|  &=  \frac{1}{c}|m + \textrm{lap}(\frac{1}{\epsilon}) - n_1c| ~ = ~ \frac{1}{c}|\textrm{bin}(n_1, c) + \textrm{lap}(\frac{1}{\epsilon}) - E[\textrm{bin}(n_1, c) + \textrm{lap}(\frac{1}{\epsilon})]|.
\end{align*}
}
In order to prove that accuracy with high probability, we use Chebyshev's inequality.
\begin{lemma}[Chebyshev's inequality]
\label{prop:cheb}
Let $X$ be a random variable with expected value $\mu$ and variance $\sigma^2$. For any real number $k > 0$,
$\Pr[|X - \mu| \geq k\sigma] \leq \frac{1}{k^2}$.
\end{lemma}

We set $k = \sqrt{3}$ and let $X \sim \textrm{bin}(n_1, c) + \textrm{lap}(\frac{1}{\epsilon})$ with $Var[X] = n_1c(1-c) + \frac{2}{\epsilon^2}$ to get
\[ \Pr\left[|\textrm{bin}(n_1, c) + \textrm{lap}(\frac{1}{\epsilon}) - E[\textrm{bin}(n_1, c) + \textrm{lap}(\frac{1}{\epsilon})]| \geq \sqrt{3(n_1c(1-c) + \frac{2}{\epsilon^2})}\right] \leq \frac{1}{3}.  \]
This is equivalent to
\[ \Pr\left[\frac{1}{c}|\textrm{bin}(n_1, c) + \textrm{lap}(\frac{1}{\epsilon}) - E[\textrm{bin}(n_1, c) + \textrm{lap}(\frac{1}{\epsilon})]| \geq \sqrt{3(\frac{n_1(1-c)}{c} + \frac{2}{\epsilon^2 c^2})}\right] \leq \frac{1}{3}.  \]

Thus, $\Pr \left[|\hat{s} - n_1| \geq \sqrt{3(\frac{n_1(1-c)}{c} + \frac{2}{\epsilon^2 c^2})} \right] \leq \frac{1}{3}$.
\end{proof}

The mechanism can pick $c$ freely. If the mechanism picks a constant $c$ such that $\frac{n(1-c)}{c} \leq \frac{2}{\epsilon^2 c^2}$, the mechanism is $O(\epsilon^{-1})$ accurate.

We will extend this result to general data entry and discrete cost distributions in Section \ref{sec:ext}. Thus, we have the main theorem.
\begin{theorem}
Mechanism 1 is BIC, EIIR, $O(\epsilon^{-1})$-accurate, perfectly data private, and $\epsilon$-differentially private.\qed
\end{theorem}

\subsection{Extensions and Computational Issues} \label{sec:ext}
\noindent{\bf General Database Entries.} Suppose that the entry of database has $d$ attributes, that is, $D_i \in [h]^d$. Given a sequence $a_1, \dots, a_d$, where $a_j \in [h]$, the data analyst wants to estimate  $|\{i : \forall_j D_{ij}  = a_j\}|$. For any $D_i$, we can transform $D_i$ to a single attribute data $D_i' = 1 + \sum_{i=0}^{d-1} D_{ij} \times d^i$, such that $D_i' \in [h^d]$. Then, we can apply the mechanism to estimate the number of players with $D_i' = 1 + \sum_{i=0}^{d-1} a_j \times d^i$.
\begin{comment}
Since the range of $D_i$ is $[h]$, in the second step of mechanism 1, we should find $\alpha_j$ for all $j \in [h]$, such that $F_j(\alpha_j) = c$. Then, the mechanism offer the contract ``If $D_i = j$, your expected payment will be $\epsilon \alpha_j$.'' By a argument similar to Lemmas \ref{lem:safe} and \ref{lem:truthful}, we can show that the mechanism is perfectly data private, BIC, and EIIR. Since every player has equal probability $c$ to accept the contract, we can show that the mechanism satisfies $\epsilon$-differential privacy of estimate and is $O(\epsilon^{-1})$-accurate by a argument similar to Lemmas \ref{lem:is-privacy} and \ref{lem:d-accuracy}. In order to satisfy differential privacy of payments, we let $\gamma := \max_j \alpha_j - \min_j \alpha_j$. For every player with $D_i = j$ who accept the contract, the payment is $\epsilon \alpha_j + \textrm{lap}(\frac{\gamma}{\epsilon})$. Then, the payments satisfy $\epsilon$-differential privacy by a argument similar to Lemma \ref{lem:ip-privacy}.
\end{comment}

\vspace{0.5em}\noindent{\bf Discrete Cost Distributions.}
When $F_j$ is a discrete probability function, the major difficulty is that for a given $c$ and $j$, we may not find a suitable $\alpha_j$, such that $F_j(\alpha_j) = c$, because the cumulative probability function of a discrete distribution is a step function. However, the mechanism can provide different contracts to different players and this ability allows us to design a mechanism for discrete case.

The basic idea is that the mechanism uses randomness to pick $\alpha_j$ such that every player has equal probability $c$ to accept the contract. For a given $c$ and for each $j$, if there is no $\alpha_j$ such that $F_j(\alpha_j) = c$, then the mechanism finds the largest $\alpha_j^-$ and the smallest $\alpha_j^+$ such that $F_j(\alpha_j^-) = c_j^- < c$ and $F_j(\alpha_j^+) = c_j^+ > c$. Note that a player $i$ with $D_i = j$ accepts the contract if his cost is smaller than the expected payment. If the expected payment is $\alpha_j^+$, then the player accepts the contract with probability $c_j^+ > c$. On the other hand, if the expected payment is $\alpha_j^-$, then the player accepts the contract with probability $c_j^- < c$. Let $\beta_j = \frac{c - c_j^-}{c_j^+ - c_j^-}$. Player $i$ is given a contract ``If $D_i = j$, your expected payment is $\epsilon \alpha_j$ in expectation," where $\Pr[\alpha_j = \alpha_j^-] = 1 - \beta_j$ and $\Pr[\alpha_j = \alpha_j^+] = \beta_j$. Thus, $\Pr[v_i \leq \alpha_j] = c_j^- + \beta_j (c_j^+ - c_j^-) = c$, where the randomness is over the distribution of costs and the random choice of $\alpha_j$. We can prove that the mechanism is perfectly data private, BIC, and EIIR by arguments similar to those in the proofs of Lemmas \ref{lem:safe} and \ref{lem:truthful}. Since every player has equal probability $c$ to accept the contract, we can show that the mechanism satisfies $\epsilon$-differential privacy of estimate and is $O(\epsilon^{-1})$-accurate by arguments similar to those in the proofs of Lemmas \ref{lem:is-privacy} and \ref{lem:d-accuracy}. In order to satisfy differential privacy of payments, we let $\gamma := \max_j \alpha_j^+ - \min_j \alpha_j^-$. Then, the payments satisfy $\epsilon$-differential privacy by an argument similar to the proof of Lemma \ref{lem:ip-privacy}.

\vspace{0.5em}\noindent{\bf Cost of Mechanism.}
For a fixed $\epsilon$, when $c$ increases, the accuracy of Mechanism is improved, since Mechanism 1 uses more players' data.
However, Mechanism 1's expected total payment also increases.
Since Mechanism 1 is $\sqrt{3(\frac{n_1(1-c)}{c}+\frac{2}{\epsilon^2c^2})}$-accurate, there is a trade-off between the accuracy and the expected total payment.
Since the mechanism can pick $c$ freely, for a given $\epsilon > \sqrt{\frac{8}{n}}$, the mechanism can pick $c = \frac{1 + \sqrt{1 - 8\epsilon^{-2}/n}}{2}$. 
Let $\alpha = \max_j \alpha_j$. 
The expected total payment is $\epsilon \alpha cn = \alpha n(\frac{\epsilon + \sqrt{\epsilon^2 - 8/n}}{2})$. Then, Mechanism 1 picks a suitable $\epsilon$, such that the expected total payment $\epsilon \alpha cn = B$.
Hence, the mechanism is budget-feasible in expectation and is $O(\frac{1}{\epsilon c}) = O(\frac{\alpha n}{B})$ accurate. 

\vspace{0.5em}\noindent{\bf Fixed Accuracy.}
If the data analyst wants a $k$-accurate mechanism, we can pick $c = \frac{1}{1 + k^2/6n}$ and $\epsilon = \frac{2\sqrt{3}(1 + k^2/6n)}{k}$, such that the mechanism is $k$-accurate. The expected total payment is $\epsilon \alpha cn = \frac{2\sqrt{3} \alpha n}{k}$.

\vspace{0.5em}\noindent{\bf Computing $\mathbf{F^{-1}(c)}$.}
In an ideal model, when $F_j$ is a continuous distribution, we assume that mechanism can access the closed form of $F_j$, such that the mechanism can compute $\alpha = F^{-1}_j(c)$ accurately. 
However, when the mechanism cannot access the closed form of $F_j$, $F^{-1}_j(c)$ may not be computable. 
When it is impossible to access the closed form of $F_j$, we assume that there is an oracle, which returns $F_j(v)$ for any given value $v$. 
In the oracle model, the mechanism finds $\alpha_j^-$, $\alpha_j^+$ for all $j$, such that $F_j(\alpha_j^-) < c$, $F_j(\alpha_j^+) > c$, and $\alpha_j^+ - \alpha_j^- < \delta$ for $\delta < 1/n$ using binary search. 
Then, the mechanism uses the method that we use for discrete cost distributions to construct the contract.
That is, let $\beta_j = \frac{c - c_j^-}{c_j^+ - c_j^-}$. 
Player $i$ is given a contract ``If $D_i = j$, your expected payment is $\epsilon \alpha_j$ in expectation," where $\Pr[\alpha_j = \alpha_j^-] = 1 - \beta_j$ and $\Pr[\alpha_j = \alpha_j^+] = \beta_j$. 
Thus, $\Pr[v_i \leq \alpha_j] = c_j^- + \beta_j (c_j^+ - c_j^-) = c$, where the randomness is over the distribution of costs and the random choice of $\alpha_j$.
Hence, the mechanism is still perfectly data private, BIC, EIIR, $\epsilon$-differential private, and $O(1/\epsilon)$-accurate. 
In the oracle model, the expected payment for player $i$ with $D_i = j$ who accepts the contract is at most $\alpha_j^+$. 
In the ideal model, the expected payment for player $i$ with $D_i = j$ who accepts the contract is $F^{-1}_j(c)$, which is smaller than $\alpha_j^+$.
Since $\alpha_j^+ - F^{-1}_j(c)$ is at most $\delta < 1/n$, the difference between the expected payments in the ideal model and in the oracle model is at most $1/n$ for each player.
Thus, the difference between the expected total payment in the ideal model and in the oracle model is at most 1.
\begin{comment}
\subsection{Negative Payments}
Since the mechanism use the Laplacian noise to make the payments $\epsilon$-differential private and the range of the Laplacian noise is over all real numbers, the payment may be negative. The players with negative payments, even if they have positive cost, should pay to the mechanism. Thus, a mechanism with negative payments is not satisfactory. However, we can construct an example that the negative payments are necessary. Suppose that every $D_i \in {0, 1\}$ and the domain of $F_0$ is all positive real numbers and the domain of $F_1$ is all negative numbers. 
\end{comment}

%% file: SellingPrivacy_Opt.tex
\section{Optimality} \label{sec:opt}
In this section, we define a benchmark for the expected payment of a mechanism and compare the expected payment of Mechanism 1 to this benchmark in two different settings.
The benchmark mechanism is not only truthful but also knows $D_i$ for all $i$ and has no privacy requirements.
We show that when all costs are non-negative, Mechanism 1 is provably close to the benchmark. 

The benchmark is the minimum expected payment among all truthful mechanisms $M^*$ that satisfy the following properties.
In order to get a meaningful estimate, for any $k < n/2$, a $k$-accurate mechanism learns a subset of players' data.
We call this subset a \emph{sample set}. 
Since obtaining an estimate based on an unbiased sample is a common approach in statistics, we assume that $M^*$ uses an unbiased sample. 
Suppose that there are $n_j$ players with $D_i = j$ for $j \in [h]$. 
Since the sample set is unbiased, there exists $c$ such that $M^*$ buys $w_j = cn_j$ data from players with $D_i = j$. 
After getting an unbiased sample, $M^*$ uses $w_1/c$ as the straightforward estimate of $n_1$.
Since the choices of $c$ may effect the accuracy guarantee, we compare the payment of Mechanism 1 to the payment of $M^*$, where Mechanism 1 and $M^*$ have the same size of sample sets. 
Thus, $M^*$ is a truthful mechanism that gets an unbiased sample with size $cn$ for a fixed number $c$.

Since there is no competition between players with data $j$ and players with data $j' \neq j$, $M^*$ can run auctions for players with $D_i = j$ for all $j \in [h]$ independently and buy $w_j$ data from players with $D_i = j$.
The mechanism that guarantees buying $w$ units is called \emph{$w$-unit procurement mechanism}.
Thus, $M^*$ is a mechanism that runs a truthful, $w_j$-unit procurement mechanism for each $j \in [h]$. 
% Thus, $M^*$ is a truthful multi-unit procurement mechanism that guarantees buying $w_j$ data for each $j \in [h]$.

Mechanism 1 buys in expectation $w_j$ data from players with $D_i = j$ for $j \in [h]$. 
We compare the expected payment of Mechanism 1 for buying in expectation $w_j$ data from players with $D_i = j$ with the expected payment of $M^*$ for buying $w_j$ data from players with $D_i = j$ for each $j$.
If the expected payment of Mechanism 1 is at most $r$ times the expected payment of $M^*$ for each $j$, then the total expected payment of Mechanism 1 is at most $r$ times the total expected payment of $M^*$.
Thus, we focus on a single auction that all players have the same $D_i$ and both Mechanism 1 and the $M^*$ want to buy $w$ data from $n$ players. 

For multi-unit procurement mechanisms, let $x_i$ be the indicator random variable denoting whether the mechanism buys from player $i$. Let $v_i$ be the cost to the player $i$, if $x_i = 1$. Let $p_i$ be the payment of player $i$. The utility for player $i$ is $p_i - x_i v_i$. Note that when we consider privacy preserving mechanisms, the utility of player $i$ is $p_i - \epsilon x_i v_i$. However, since $\epsilon$ is the same for all players, we can ignore $\epsilon$ without loss of generality, that is, scaling every player's cost by $\epsilon$. Without loss of generality, we suppose that players report costs $v_1 \leq v_2 \dots \leq v_n$.

\begin{comment}
First, we show that Mechanism 1 is 2-approximate to the optimal truthful, EPIR, and envy-free mechanism. Next, we show that Mechanism 1 is 2-approximate to the optimal truthful mechanism, when $F$ is anti-regular. Finally, for distribution satisfies Assumption 1 but $\phi(z)$ is not increasing in $z$, we show that Mechanism 1 is $2r$-approximate to the optimal truthful mechanism, where $r$ is a number depending on the distribution.
\end{comment}

\subsection{Envy-free Benchmark} 
A mechanism is \emph{envy-free} if for all $v$ and for all $i$, $j$, $p_i - v_i x_i \geq p_j - v_i x_j$. 
We show that for any envy-free, multi-unit procurement mechanism, every data that is bought by the mechanism  is purchased at the same price.
Suppose that a multi-unit procurement mechanism buys data from two players at two different prices.
Since the player with the lower price wants to have the higher price, the mechanism is not envy-free.
We compare the expected payment of Mechanism 1 with the expected payment of the optimal, envy-free, dominant strategy truthful, multi-unit procurement mechanism. 
We use envy-free mechanisms as a benchmark, because for procurement mechanisms in a Bayesian setting, the optimal mechanisms are known to charge a fixed price.\footnote{Envy-free benchmarks are also common in prior-free mechanism design~\cite{Hart08}.} 

We introduce another commonly used solution concept as follows.
A profile of strategies $q_1, \dots, q_n$ is a \emph{dominant strategy equilibrium} if for all $i$, $v_i$,$v_{-i}$, and $y_i' \in Y$, $E[u_i(q(v_i, v_{-i}, D), v_i)] \geq E[u_i((y_i', q_{-i}(v_{-i}, D_{-i})), v_i)]$, where the randomness is from the mechanism. 
A direct mechanism is \emph{dominant strategy truthful} if $q_i(v_i, D_i) = v_i$ is a dominant strategy equilibrium for every player $i$. 
The following lemma characterizes the total payment for any dominant strategy truthful, EPIR, and envy-free mechanisms.

\begin{lemma}[Theorem 4.6 in~\cite{Roth11}]\label{prop:p-lower}
No dominant strategy truthful, EPIR, and envy-free $w$-unit procurement mechanism can have total payment less than $w v_{w+1}$. \qed
\end{lemma}

Let $F$ be the cumulative distribution function of players' costs, that is, $F(a) = \Pr[v \leq a]$. By Lemma \ref{prop:p-lower}, the total expected payment of any dominant strategy truthful, EPIR, and envy-free $w$-unit procurement mechanism is at least $wE_{v \sim F}[v_{w+1}]$. Thus, our benchmark is $wE_{v \sim F}[v_{w+1}]$.

Now, we compare the benchmark with the expected payment of Mechanism 1. There are two cases. First, when there exists $\alpha$ such that $F(\alpha) = \frac{w}{n}$, Mechanism 1 offers a posted price $\alpha$ for each player in order to buy $w$ players' data in expectation. If player $i$ accepts the price, the mechanism buys from player $i$ with expected payment $\alpha$. Since each player has probability $\frac{w}{n}$ to accept the contract, the total expected payment of Mechanism 1 is $w\alpha$. 

Second, when there is no $\alpha$ such that $F(\alpha) = \frac{w}{n}$, we give an extension to Mechanism 1 in Section \ref{sec:ext}. The extension finds the largest $\alpha^-$ and the smallest $\alpha^+$, such that $F(\alpha^-) < \frac{w}{n}$ and $F(\alpha^+) > \frac{w}{n}$. Let $c^- := F(\alpha^-)$, $c^+ := F(\alpha^+)$, and $\beta := \frac{\frac{w}{n} - c^-}{c^+ - c^-}$. Then, the mechanism offers a price $\alpha^+$ with probability $\beta$ and price $\alpha^-$ with probability $1 - \beta$. For a player with cost at most $\alpha^-$, since the player always accepts the offer, the expected payment is $(\alpha^- (1-\beta) + \alpha^+ \beta)$. For a player with cost equal to $\alpha^+$, since the player accepts the offer only when the offered price is $\alpha^+$, the expected payment is $\alpha^+\beta$. For a player with cost larger than $\alpha^+$, since the player always rejects the offer, the expected payment is 0. Since each player has a cost at most $\alpha^-$ with probability $c^-$ and has a cost equal to $\alpha^+$ with probability $c^+ - c^-$, each player's expected payment is $c^- (\alpha^-(1-\beta) + \alpha^+ \beta) + (c^+ - c^-)\alpha^+ \beta$. Thus, the total expected payment is $n (c^- (\alpha^-(1-\beta) + \alpha^+ \beta) + (c^+ - c^-)\alpha^+ \beta)$ by the linearity of expectation. Moreover, 
{
\allowdisplaybreaks
\begin{align*}
n (c^- (\alpha^-(1-\beta) + \alpha^+ \beta) + (c^+ - c^-)\alpha^+ \beta) ~&=~ n(c^- (\alpha^-(1-\beta) + \alpha^+ \beta) + (\frac{w}{n} - c^-)\alpha^+)\\
                                             &=~ n(\frac{w}{n} \alpha^+ + c^-(\alpha^-(1-\beta) + \alpha^+ \beta - \alpha^+)) \\
                                             &=~ n(\frac{w}{n} \alpha^+ + c^-((1 - \beta)(\alpha^- - \alpha^+))\\
                                             &=~ n(\frac{w}{n} \alpha^+ - c^-((1 - \beta)(\alpha^+ - \alpha^-)))\\
                                             &=~ w(\alpha^+ - \frac{nc^-}{w}(1-\beta)(\alpha^+ - \alpha^-)).
\end{align*}
}

When there exists $\alpha$, such that $F(\alpha) = \frac{w}{n}$, the expected payment of Mechanism 1 is $w\alpha$. When $\alpha$ does not exist, the expected payment is $w(\alpha^+ - \frac{nc^-}{w}(1-\beta)(\alpha^+ - \alpha^-))$. Thus, we should compare both $w\alpha$ and $w(\alpha^+ - \frac{nc^-}{w}(1-\beta)(\alpha^+ - \alpha^-))$ with $wE_{v \sim F}[v_{w+1}]$. It suffices to compare $\alpha$ and $\alpha^+ - \frac{nc^-}{w}(1-\beta)(\alpha^+ - \alpha^-)$ with $E_{v \sim F}[v_{w+1}]$. 

\begin{lemma}\label{lem:2app}
1. If there exists $\alpha$ such that $F(\alpha) = \frac{w}{n}$, then $E_{v \sim F}[v_{w+1}] \geq  \frac{1}{2} \alpha$. \\
2. If there is no $\alpha$ such that $F(\alpha) = \frac{w}{n}$, then $E_{v \sim F}[v_{w+1}] \geq \frac{1}{2}(\alpha^+ - \frac{n}{w}c^-(1-\beta)(\alpha^+ - \alpha^-))$. 
\end{lemma}
\begin{proof}
We show the second statement. The first statement follows by setting $\alpha^- = \alpha^+ = \alpha$.

Let $\eta = \alpha^+ - \frac{nc^-}{w}(1-\beta)(\alpha^+ - \alpha^-)$. By conditional probability, 
\begin{align*}
E[v_{w+1}] &= \Pr[v_{w+1} \leq \eta] \times E[v_{w+1} \mid v_{w+1} \leq \eta] + \Pr[v_{w+1} > \eta] \times E[v_{w+1} \mid v_{w+1} > \eta] \\
           &\geq \Pr[v_{w+1} > \eta] \times \eta \text{ \     \ (costs are non-negative).}
\end{align*}
It suffices to show that $\Pr[v_{w+1} > \eta] \geq \frac{1}{2}$. Since $c^- < \frac{w}{n}$, $\frac{nc^-}{w} < 1$. Since $\beta < 1$ and $\frac{nc^-}{w} < 1$, $\alpha^- < \eta < \alpha^+$. If $v_{w+1} > \eta$, then $v_{w+1} \geq \alpha^+$, since $\alpha^+$ is the smallest number larger than $\alpha^-$ with non-zero probability. Let $v_{(i)}$ denote the cost of player $i$. If $v_{w+1} \geq \alpha^+$, then at most $w$ players' $v_{(i)}$ are no larger than $\alpha^-$. Since each $v_{(i)}$ is independently drawn from $F$, $\Pr[v_{(i)} \leq \alpha^-] = c^-$. Let $X_i$ be the indicator random variable such that $X_i = 1$ if $v_{(i)} \leq \alpha^-$, otherwise $X_i = 0$. Let $X = \sum_{i=1}^n X_i$. The probability that at most $w$ players have $v_{(i)}$ no larger than $\alpha^-$ is $\Pr[X \leq w]$. Since the $X_i$'s are independent, identical, indicator random variables, $X$ is a random variable from a binomial distribution $\textrm{Bin}(n, c^-)$. Thus, $\Pr[v_{w+1} \geq \alpha^+] = \Pr[\textrm{bin}(n, c^-) \leq w]$.

Now, we show that $\Pr[\textrm{bin}(n, c^-) \leq w] \geq \frac{1}{2}$. We say $m$ is the \emph{median} of a distribution $D$ over real numbers if, $\Pr[Z \leq m ] \geq \frac{1}{2}$ and $\Pr[Z \geq m] \geq \frac{1}{2}$, where $Z$ is a random variable drawn from $D$. For a binomial distribution $\textrm{Bin}(n, p)$, the expected value $np$ and the median $m$ satisfy $\lfloor np \rfloor \leq m \leq \lceil np \rceil$~\cite{Kaas80}. Since $c^- < \frac{w}{n}$, the expected value of $\textrm{bin}(n, c^-)$ is smaller than $w$. Since $\lceil nc^- \rceil \leq w$, the median $m$ of $\textrm{Bin}(n, c^-)$ is at most $w$. Thus, $\Pr[\textrm{bin}(n, c^-) \leq w] \geq \frac{1}{2}$.
\end{proof}
%By~\cite{Nick10}, if the mean of a binomial distribution is an integer, then the median equals the mean. Thus, $w$ is the median for binomial distribution $\textrm{Bin}(n, \frac{w}{n})$. Therefore, $\Pr[\textrm{bin}(n, \frac{w}{n}) \leq w] \geq \frac{1}{2}$. 

Lemmas \ref{prop:p-lower} and \ref{lem:2app} imply the following theorem.
\begin{theorem} \label{thm:envy}
Mechanism 1's expected payment is 2-approximate to the benchmark. \qed
\end{theorem}

\begin{comment}
We also can construct an example to show that this analysis is tight. Let $w = n/2$ and $F$ be a distribution over $\{0, 2\}$, where $\Pr[v_i = 0] = \frac{1}{2}$ and $\Pr[v_i = 2] = \frac{1}{2}$. Thus, the mechanism. By conditional probability, 
\begin{align*}
E[v_{w+1}] &= \Pr[v_{w+1} < F^{-1}(\frac{w}{n})] \times E[v_{w+1} \mid v_{w+1} < F^{-1}(\frac{w}{n})] \\
           & \hspace{0.5cm} + \Pr[v_{w+1} \geq F^{-1}(\frac{w}{n})] \times E[v_{w+1} \mid v_{w+1} \geq F^{-1}(\frac{w}{n})] \\
           &= 2\Pr[v_{w+1} \geq F^{-1}(\frac{w}{n})]. 
\end{align*}

When $n \rightarrow \infty$, $\Pr[v_{w+1} \geq F^{-1}(\frac{w}{n})]$ approaches to 0.5 by approximating the binomial distribution using a normal distribution. Thus, when $n \rightarrow \infty$, $E[v_{w+1}] = 1$, which is half of $F^{-1}(\frac{w}{n})$.
\end{comment}

\subsection{Anti-regular Distributions}
In this section, we compare the expected payment of Mechanism 1 with the expected payment of the optimal BIC, multi-unit procurement mechanism.
We first characterize randomized BIC procurement mechanisms.
For a randomized mechanism and a given bid $v_i$, let $\bar{x}_i(v_i)$ be the probability that the mechanism buys from player $i$ and let $p_i(v_i)$ be the random variable denoting the payment for player $i$, where both $\bar{x}_i$ and $p_i$'s randomness come from the mechanism and $v_{-i}$.
Suppose that when $v_i = \infty$, the mechanism will not buy from player $i$. 
That is, $\bar{x}_i(\infty) = 0$ and $E[p_i(\infty)] = 0$. 
The characterization for the BIC, procurement mechanisms is analogous to the characterization of BIC selling mechanisms, which is a well-known result in auction theory. 
We provide a proof of the following characterization in the Appendix.
\begin{lemma}\label{prop:identity}
A randomized procurement mechanism is BIC if and only if for every $i$ the procurement probability $\bar{x}$ and payment $p$ satisfies \\
(i) $\bar{x}_i(v_i)$ is decreasing in $v_i$;\\
(ii) $E[p_i(v_i)] = v_i \bar{x}_i(v_i) + \int_{v_i}^{\infty} \bar{x}_i(t)dt$. \qed
\end{lemma}

To prove the optimality of selling mechanisms, Myerson~\cite{My81} introduces a virtual value function. The analogous function for procurement mechanisms is a \emph{virtual cost function}, which is $\phi(z) := z + \frac{F(z)}{f(z)}$. Thus, to ensure that $\phi(z)$ is well-defined and the integral of $f$ is well-defined (used in the proof of Lemma \ref{prop:w-opt} and Lemma \ref{prop:g-opt}), we assume 

\begin{asm}
Let $f$ be the density probability function of distribution $F$ with range $[a, b] \subseteq [0, \infty)$. $f$ is piecewise continuous and $f(z)$ is positive for all $z \in [a, b]$.
\end{asm}

A distribution $F$ is \emph{anti-regular} if $F$ satisfies Assumption 1 and $\phi(z)$ is increasing in $z$.\footnote{For selling mechanisms, a distribution is \emph{regular} if the \emph{virtual value} $\phi'(z) = z - \frac{1-F(z)}{f(z)}$ is increasing in $z$.} 

When the distribution $F$ is anti-regular, ~\cite{Elkind04} characterize the optimal dominant strategy truthful mechanism to minimize the expected payment for path auctions. Although their problem is not exactly the same as $w$-unit procurement mechanisms, their result can be extended to procurement mechanisms easily. For completeness, we provide a proof of the following lemma for our setting in the Appendix.

%When the distribution $F$ satisfies Assumption 1,~\cite{IK08} characterize the optimal truthful procurement mechanism to minimize the expected payment in a more general Bayesian setting where players can sell more than one item and players can lie about the number of items they have. Their model includes our setting in which every player can sell at most one item. However, their proof is complex, since they deal with a general setting. For completeness, we provide a simpler proof of the following lemma for our setting in the Appendix.

\begin{lemma}\label{prop:w-opt}
When the distribution $F$ is anti-regular, the optimal BIC $w$-unit procurement buys from the $w$ players with the smallest virtual cost. \qed
\end{lemma}

Since $\phi(z)$ is increasing in $z$, the optimal mechanism buys from the first $w$ players. 
By Lemma \ref{prop:identity}, the expected payment for player $i \leq w$ is $v_{w+1}$. 
Thus, the total expected payment of the optimal BIC mechanism is $w E_{v \sim F}[v_{w+1}]$. Thus, our benchmark is $wE_{v \sim F}[v_{w+1}]$. We compare the expected payment of Mechanism 1 with $wE_{v \sim F}[v_{w+1}]$, when $F$ is anti-regular.

\begin{theorem}
When $F$ is anti-regular, Mechanism 1's expected payment is 2-approximate to the benchmark.
\end{theorem}

\begin{proof}
Since $F$ satisfies Assumption 1 by definition of anti-regular, $F^{-1}$ is well-defined. The total expected payment of Mechanism 1 is $wF^{-1}(\frac{w}{n})$. When $F$ is anti-regular, the benchmark is $w E_{v \sim F}[v_{w+1}]$. By Lemma \ref{lem:2app}, Mechanism 1 is 2-approximate.
\end{proof}

\subsection{General Distributions} When the distribution satisfies Assumption 1 but $\phi(z)$ is not increasing in $z$, buying from the $w$ players with smallest virtual cost may result in a non-truthful mechanism. We can use the ironing procedure, which is designed by Myerson~\cite{My81}, to resolve this issue. For a fixed cost vector $v$, ironing procedure irons on interval $[a, b)$, if $v_i \in [a, b)$, then $v_i$ is replaced by a random number $v_i'$, which is drawn from the distribution $F$ on $[a, b)$. By a way similar to Myerson's method, we can identify a set $S$ of intervals, such that the \emph{ironed virtual cost function} $\bar{\phi}(z) = E[\phi(z)]$ is increasing in $z$. Moreover, for an ironed interval $[a, b)$, $\bar{\phi}(z)$ is the same for all $z \in [a, b)$. 
The formal definitions of the ironed interval set $S$ and ironed virtual cost function are in the appendix.

\begin{lemma} \label{prop:g-opt}
The $w$-unit procurement mechanism that buys from the $w$ players with smallest ironed virtual cost and breaks ties uniformly at random is the optimal BIC mechanism when the distribution satisfies Assumption 1. \qed
\end{lemma}

Thus, our benchmark is the expected payment of the optimal BIC mechanism, $M$, when the distribution satisfies Assumption 1. In order to calculate the expected payment of $M$, we specify the payment rule as follows. Let $\bar{x}_i(v_i, v_{-i})$ be the probability that $M$ buys from player $i$, where the randomness comes from the mechanism. Since $M$ buys from the $w$ players with smallest ironed virtual cost,  $\bar{x}_i(v_i, v_{-i})$ is decreasing in $v_i$ for any fixed $v_{-i}$. Let $p_i(v_i, v_{-i})$ be the random variable denoting the payment for player $i$, where $E[p_i(v_i, v_{-i})] = v_i\bar{x}_i(v_i, v_{-i}) + \int_{v_i}^{\infty} \bar{x}_i(t, v_{-i})dt$ and the randomness comes from the mechanism. It is easy to see that this payment rule satisfies Lemma \ref{prop:identity}.

We compare the expected payment of Mechanism 1 with the benchmark.

\begin{theorem}
Let $F$ satisfy Assumption 1. Let $S$ be the set of ironed intervals for $F$. If every interval $[a, b) \in S$ satisfies $a \geq b/r$ for some $r > 1$, then the expected payment of Mechanism 1 is $2r$-approximate to the benchmark.
\end{theorem}
\begin{proof}
Since $F$ satisfies Assumption 1, $F^{-1}$ is well-defined. The expected payment of Mechanism 1 is $wF^{-1}(\frac{w}{n})$. We compare the expected payment of the optimal BIC mechanism, $M$, with $wF^{-1}(\frac{w}{n})$. Let $p_i(v)$ be the random variable representing the payment for player $i$ in $M$ when the cost vector is $v$. We show that $E_{v\sim F, M}[\sum_{i \in [n]} p_i(v)] \geq E_{v \sim F}[v_{w+1}]/r$, which implies $E_{v\sim F, M}[\sum_{i \in [n]} p_i(v)] \geq F^{-1}(\frac{w}{n})/2r$ by Lemma \ref{lem:2app} and hence Mechanism 1 is $2r$-approximate.

There are two sources of randomness in mechanism $M$. One is from the cost vector $v$ since $v$ is drawn from a distribution $F$. Another one is $M$ itself since $M$ is a randomized mechanism. For a fixed cost vector $v$, let $p_i^v$ be the random variable representing the payment for player $i$, where the randomness only comes from $M$. We show that for any fixed $v$, $E_{M}[\sum_{i \in [n]} p_i^v] \geq wv_{w+1} / r$. This implies $E_{v\sim F, M}[\sum_{i \in [n]} p_i(v)] \geq E_{v \sim F}[v_{w+1}]/r$. There are three cases.  \\
\textbf{Case 1:} $v_{w+1}$ is not in any ironed interval. Since $M$ chooses the $w$ players with smallest ironed virtual costs and the ironed virtual cost is increasing, $M$ buys from the first $w$ players. For  player $i \leq w$, if $v_i$ increases to $t < v_{w+1}$, by the monotonicity of $\bar{\phi}$, the mechanism still buys from player $i$. That is $\bar{x}_i(t, v_{-i}) = 1$ for all $t < v_{w+1}$. When $t > v_{w+1}$, the mechanism will not buy from player $i$. Thus, by definition of the expected payment, the expected payment for each player $i \leq w$ is $v_i + \int_{v_i}^{\infty} \bar{x}_i(t, v_{-i}) dt = v_i + \int_{v_i}^{v_{w+1}} \bar{x}_i(t, v_{-i}) dt = v_{w+1}$. Since expected payment for player $i > w$ is 0, $E_M[\sum_{i \in [n]} p_i^v] = wv_{w+1}$.\\
\textbf{Case 2:} $v_{w+1}$ is in an ironed interval $[a, b)$ but $v_w \notin [a, b)$. Since for all player $i \leq w$, $v_i \notin [a, b)$, $M$ buys from the first $w$ players. For player $i \leq w$, $\bar{x}_i(t, v_{-i}) = 1$ for all $t < a$. By definition of the expected payment, the expected payment for player $i \leq w$ is $v_i + \int_{v_i}^{\infty} \bar{x}_i(t, v_{-i}) dt \geq v_i + \int_{v_i}^{a} \bar{x}_i(t, v_{-i}) dt = a$. Thus, $E_M[\sum_{i \in [n]} p_i^v] \geq wa \geq wb/r \geq wv_{w+1} / r$. \\
\textbf{Case 3:} $v_{w+1}$ and $v_{w}$ are in the same ironed interval $[a, b)$. Let $l_1 = |\{i: v_i < a\}|$ and $l_2 = |\{i: v_i \in [a, b)\}|$. Thus, $l_1 < w$ and $l_1 + l_2 > w$. The mechanism always buys from the first $l_1$ players. Since $\bar{\phi}(t)$ is the same for all $t \in [a, b)$ and the mechanism breaks ties uniformly at random, the mechanism buys from player $i$, $l_1 + 1 \leq i \leq l_1 + l_2$, with probability $\frac{w - l_1}{l_2}$. For player $i \leq l_1$, $\bar{x}_i(t, v_{-i}) = 1$ if $t < a$. By definition of the expected payment, the expected payment for player $i \leq l_1$ is $v_i + \int_{v_i}^{\infty} \bar{x}_i(t, v_{-i}) dt \geq v_i + \int_{v_i}^{a} \bar{x}_i(t, v_{-i}) dt = a$. For player $i$, $l_1 < i \leq l_1 + l_2$, when $v_i$ increases to $t < b$, since $\bar{\phi}(t)$ is the same for all $t \in [a, b)$, the probability that the mechanism buys from player $i$ does not change. That is, $\bar{x}_i(t, v_{-i}) = \frac{w - l_1}{l_2}$ if $t \in [a, b)$. By definition of the expected payment, the expected payment for player $i$, $l_1 < i \leq l_1 + l_2$, is $v_i\bar{x}_i(v_i, v_{-i}) + \int_{v_i}^{\infty} \bar{x}_i(t, v_{-i}) dt = \frac{b(w - l_1)}{l_2}$. Therefore, $E_M[\sum_{i \in [n]} p_i^v] \geq al_1 + l_2 \frac{b(w - l_1)}{l_2} \geq wa \geq wb/r \geq wv_{w+1} / r$.
\end{proof}

%% file: SellingPrivacy_App.tex
\appendix

\section{Optimal Procurement Mechanisms}
We first characterize the BIC randomized procurement mechanisms in a way similar to Myerson's characterization of truthful selling mechanisms~\cite{My81}~\cite{AGT}. We assume $\bar{x}_i(\infty) = 0$ and $E[p_i(\infty)] = 0$.
\begin{lemma}[Lemma \textbf{\ref{prop:identity}}]
A randomized procurement mechanism is BIC if and only if for every $i$ the procurement probability $\bar{x}$ and payment $p$ satisfies \\
(i) $\bar{x}_i(v_i)$ is decreasing in $v_i$;\\
(ii) $E[p_i(v_i)] = v_i \bar{x}_i(v_i) + \int_{v_i}^{\infty} \bar{x}_i(t)dt$. 
\end{lemma}
\begin{proof}
($\rightarrow$) We need to show that for all $v_i'$, $E[p_i(v_i)] - v_i\bar{x}_i(v_i) \geq E[p_i(v_i')] - v_i\bar{x}_i(v_i')$. 
By (ii), it is equal to show $\int_{v_i}^{\infty} \bar{x}_i(t)dt \geq \int_{v_i'}^{\infty} \bar{x}_i(t)dt + (v_i' - v_i)\bar{x}_i(v_i')$. 
If $v_i' > v_i$, then it equals $\int_{v_i}^{v_i'} \bar{x}_i(t)dt \geq (v_i' - v_i)\bar{x}_i(v_i')$, which is true due to the monotonicity of $\bar{x}_i$. 
If $v_i' < v_i$, it equals $(v_i - v_i')\bar{x}_i(v_i') \geq \int_{v_i'}^{v_i} \bar{x}_i(t)dt$, which is true due to the monotonicity of $\bar{x}_i$.

($\leftarrow$) Since the mechanism is BIC, for all $v_i$ and $v_i'$, $E[p_i(v_i)] - v_i\bar{x}_i(v_i) \geq E[p_i(v_i')] - v_i\bar{x}_i(v_i')$. 
Symmetrically, we have $E[p_i(v_i)] - v_i'\bar{x}_i(v_i) \leq E[p_i(v_i')] - v_i'\bar{x}_i(v_i')$. 
By subtracting the inequalities, we get $(v_i' - v_i)\bar{x}_i(v_i) \geq (v_i' - v_i)\bar{x}_i(v_i')$, which implies (i). 
By rearranging these two inequalities, we get $v_i'(\bar{x}_i(v_i) - \bar{x}_i(v_i')) \geq E[p_i(v_i)] - E[p_i(v_i')] \geq v_i(\bar{x}_i(v_i) - \bar{x}_i(v_i'))$. 
Let $v_i' = v_i + \epsilon$, and divide all by $\epsilon$. 
When $\epsilon \rightarrow 0$, both sides have the same value. 
Thus, we get $v \frac{d\bar{x}_i(v_i)}{dv_i} = \frac{dE[p_i(v_i)]}{dv_i}$. 
Since $\bar{x}_i(\infty) = 0$ implies $E[p_i(\infty)] = 0$, we have $p_i(v_i) = \int_{\infty}^{v_i} v \bar{x}_i'(v_i)dv$. 
Applying integration by parts, we can get (ii).
\end{proof}

When the cost of players are drawn from a publicly known distribution $F$, we characterize the optimal BIC mechanism to minimize the payment, when $F$ is anti-regular. In~\cite{My81}, Myerson characterizes the optimal BIC mechanism to maximize the revenue for selling mechanisms assuming the distribution is regular. 
The proof of Lemma \ref{prop:2} follows the proof of Myersons's characterization~\cite{AGT}.

\begin{lemma} \label{prop:2} [Lemma \textbf{\ref{prop:w-opt}}]
When the distribution $F$ is anti-regular, the optimal BIC $w$-unit procurement buys from the $w$ players with the smallest virtual cost.
\end{lemma}
\begin{proof}
Let $\phi(z) = z + \frac{F(z)}{f(z)}$. Suppose that for any BIC mechanism, the expected payment is equal to its expected virtual cost, that is $E_{v \sim F}[\sum_{i \in [n]} p_i(v)] = E_{v \sim F}[\sum_{i \in [n]} \phi(v_i)x_i(v)]$. This implies that if the mechanism buys from $w$ players the with the smallest virtual cost, then the mechanism minimizes the payment. Moreover, since $F$ is anti-regular, $\phi(z) = z + \frac{F(z)}{f(z)}$ is increasing in $z$. Since the mechanism buys from $w$ players with smallest virtual cost, $\bar{x}_i(v_i)$ is decreasing in $v_i$ for all $i$. Hence, the mechanism is BIC. Thus, it suffices to show that $E_{v \sim F}[\sum_{i \in [n]} p_i(v)] = E_{v \sim F}[\sum_{i \in [n]} \phi(v_i)x_i(v)]$.

In order to show that the expected payment is equal to its expected virtual cost, it suffices to show that the expected payment of player $i$ is $E_{v \sim F}[\phi(v_i)\bar{x}_i(v_i)]$, since each $v_i$ is drawn from $F$ independently.

\begin{lemma}\label{lem:vir}
The expected payment of player $i$ is $E_{v \sim F}[\phi(v_i)\bar{x}_i(v_i)]$. 
\end{lemma}
\begin{proof}
Since the density function $f$ is piecewise continuous, there exists a partition $[a_1, b_1], \dots, [a_h, b_h]$ of $f$'s domain, such that $f$ is continuous within every interval $[a_i, b_i]$. Note that $b_i = a_{i+1}$ for all $1 \leq i \leq h - 1$.
{
\allowdisplaybreaks
\begin{align*}
E_{v}[p_i(v_i)] &= \sum_{j=1}^h \left(\int_{a_j}^{b_j} E[p_i(v_i)]f(v_i)dv_i \right)  \\
 &= \sum_{j=1}^h \left(\int_{a_j}^{b_j} v_i\bar{x}_i(v_i)f(v_i)dv_i + \int_{a_j}^{b_j} \int_{v_i}^{b_h} \bar{x}_i(z)f(v_i)dzdv_i\right)\text{\tag*{(Lemma \ref{prop:identity})}}\\
 &= \sum_{j=1}^h \left(\int_{a_j}^{b_j} v_i\bar{x}_i(v_i)f(v_i)dv_i \right. \\
& \text{\hspace{0.5cm}} \left. +\int_{a_j}^{b_j} \bar{x}_i(z) \int_{a_j}^{z} f(v_i)dv_i dz + \int_{b_j}^{b_h} \bar{x}_i(z) \int_{a_j}^{b_j} f(v_i)dv_i dz \right) \text{\tag*{(switch the order of integration)}}\\
 &= \sum_{j=1}^h \left(\int_{a_j}^{b_j} v_i\bar{x}_i(v_i)f(v_i)dv_i \right. \\
& \text{\hspace{0.5cm}} \left. +\int_{a_j}^{b_j} \bar{x}_i(z) (F(z) - F(a_j))dz + \int_{b_j}^{b_h} \bar{x}_i(z) (F(b_j) - F(a_j)) dz \right) \\
&= \sum_{j=1}^h \left(\int_{a_j}^{b_j} \bar{x}_i(v_i)(v_if(v_i) + F(v_i))dv_i \right. \\
& \text{\hspace{0.5cm}} \left. - F(a_j)\int_{a_j}^{b_j} \bar{x}_i(v_i) dv_i + (F(b_j) - F(a_j))\int_{b_j}^{b_h} \bar{x}_i(v_i) dv_i \right) \\
&= \sum_{j=1}^h \left(\int_{a_j}^{b_j} \bar{x}_i(v_i)(v_if(v_i) + F(v_i))dv_i - F(a_j)\int_{a_j}^{b_h} \bar{x}_i(v_i) dv_i + F(b_j)\int_{b_j}^{b_h} \bar{x}_i(v_i) dv_i \right) \\
&= \sum_{j=1}^h \left(\int_{a_j}^{b_j} \bar{x}_i(v_i)(v_if(v_i) + F(v_i))dv_i \right) \\
& \text{\hspace{0.5cm}}  - \sum_{j=1}^h \left( F(a_j) \int_{a_j}^{b_h} \bar{x}_i(v_i) dv_i - F(b_j)\int_{b_j}^{b_h} \bar{x}_i(v_i)dv_i \right) \\
&= \sum_{j=1}^h \left(\int_{a_j}^{b_j} \bar{x}_i(v_i)(v_i + \frac{F(v_i)}{f(v_i)})f(v_i)dv_i\right) \\
& \text{\hspace{0.5cm}} - \sum_{j=1}^h \left(F(a_j)\int_{a_j}^{b_h} \bar{x}_i(v_i) dv_i - F(b_j)\int_{b_j}^{b_h} \bar{x}_i(v_i) dv_i \right) \\
 &= E_{v_i}[\phi(v_i) \bar{x}_i(v_i)] \text{\tag*{\hspace{1cm}  ($F(a_1) = 0$, $b_i = a_{i+1}$ for all $i \leq i \leq h-1$) \qedhere}}
\end{align*}
}
\end{proof}
\hfill (End of proof of Lemma \ref{prop:2}) \qedhere
\end{proof}

Now, we consider the case that $F$ satisfies Assumption 1 but $\phi(z)$ is not monotone in $z$. For selling mechanisms, Myerson~\cite{My81} designs an ironing procedure to get the optimal BIC mechanism to maximize the revenue when $F$ satisfies Assumption 1. We show how to iron virtual values in the setting of procurement mechanism and use this to design an optimal BIC mechanism to minimize the payment.

Suppose that the $\phi(z)$ is not monotone. We want to transform $\phi(z)$ to another function $\bar{\phi}(z)$, such that $\bar{\phi}(z)$ is increasing in $z$. Let $q = F(v)$ and $h(q) = \phi(F^{-1}(q))$. Since the density function $f$ is always positive, $F$ is a strictly increasing. Thus, $\phi(z)$ is increasing in $z$ if and only if $h(q)$ is increasing in $q$. Moreover, $h(q)$ is increasing in $q$ if and only if $H(q) = \int_0^q h(t)dt$ is convex. However, $H$ is not convex, since $\phi(z)$ is not monotone. Thus, we want to modify $H$ to get a convex function $G$ and define $\bar{\phi}(z)$ based on $G$.

Let $S$ be the \emph{epigraph} of $H$, that is $S = \{(q, y) \mid y \geq H(q) \}$. 
Geometrically, if we draw $y = H(q)$ on a plane, then $S$ is the area containing $H$ and above $H$. 
Let conv ($S$) denote the convex hull of set $S$. 
The \emph{convex hull} of $H(q)$ is $G(q) = \min\{y \mid (q, y) \in \textrm{conv ($S$)} \}$ (Chapter 5 in~\cite{Rock97}). 
Geometrically, if we draw $y = G(q)$ on a plane, then $G$ is the lower boundary of conv ($S$). 
By definition, a function is convex if its epigraph is a convex set. 
Since the epigraph of $G$, conv ($S$), is a convex set, $G$ is convex. 
Since $G$ is the lower boundary of conv ($S$), $G(q) \leq H(q)$ for all $q \in [0, 1]$.

We define the ironed interval set and $\bar{\phi}(z)$ as follows.
%Mathematically, $G(q) := \min \{ \omega H(a) + (1-\omega)H(b) \mid \omega, a, b \in [0, 1], \, \omega a + (1-\omega) b = q\}$. 
%Note that $G(q) \leq H(q)$ for all $q \in [0, 1]$, since $G(q) \leq \omega H(q) + (1-\omega)H(q)$ for all $\omega \in [0, 1]$. 
Let $T$ be the set of points that $H(q)$ and $G(q)$ differ, that is, $T = \{q \mid  H(q) \neq G(q)\}$. 
Let $S$ be the smallest set of intervals $[y_i, z_i)$, such that $T = \cup_i (y_i, z_i)$. 
The \emph{ironed interval set} is defined as $\{\, [F^{-1}(y_i), F^{-1}(z_i)) \mid [y_i, z_i) \in S \}$.
Since $G$ is convex, $G$ is differentiable on a dense subset of $[0, 1]$ by Theorem 25.5 in \cite{Rock97}.
We define $g(q) := \frac{dG}{dq}(q)$, whenever $\frac{dG}{dq}(q)$ is well-defined, and extend $g$ to $[0, 1]$ by right-continuity.
The \emph{ironed virtual cost} function is defined as $\bar{\phi}(z) = g(F(z))$.

\begin{lemma}[Lemma \textbf{\ref{prop:g-opt}}]
The $w$-unit procurement mechanism that buys from the $w$ players with smallest ironed virtual cost and breaks ties uniformly at random is the optimal BIC mechanism when the distribution satisfies Assumption 1.
\end{lemma}
\begin{proof}
Since $G$ is a convex function, $g(q)$ is increasing in $q$. Thus, $\bar{\phi}(z)$ is increasing in $z$. Since $\bar{\phi}(z)$ is increasing in $z$ and the mechanism buys from the $w$ players with smallest ironed virtual cost, the mechanism is BIC. We only need to show that the mechanism minimizes the payment. First, we want to relate the mechanism's payment to $\bar{\phi}(z)$. Since the density function $f$ is piecewise continuous, there exists a partition $[a_1, b_1], \dots, [a_h, b_h]$ of $f$'s domain, such that $f$ is continuous within every interval $[a_i, b_i]$. Note that $b_i = a_{i+1}$ for all $1 \leq i \leq h - 1$. For any BIC mechanism, $\bar{x}_i(v_i)$ is decreasing in $v_i$ by Lemma \ref{prop:identity}. For a fixed $\bar{x}_i$,
{
\allowdisplaybreaks
\begin{align*}
E_{v \sim F}[p_i( v_i)] &= E_{v}[\phi(v_i)\bar{x}_i(v_i)] \text{\tag*{(Lemma \ref{lem:vir})}}\\
 &= E_{v}[\bar{\phi}(v_i)\bar{x}_i(v_i)] - E_{v}[(\bar{\phi}(v_i)-\phi(v_i))\bar{x}_i(v_i)] \\
 &= E_{v}[\bar{\phi}(v_i)\bar{x}_i(v_i)] - \sum_{j=1}^h \int_{a_j}^{b_j} (\bar{\phi}(v_i) - \phi(v_i))\bar{x}_i(v_i)f(v_i) dv_i \\
 &= E_{v}[\bar{\phi}(v_i)\bar{x}_i(v_i)] - \sum_{j=1}^h \int_{a_j}^{b_j} (g(F(v_i)) - h(F(v_i)))\bar{x}_i(v_i)f(v_i) dv_i \\
 &= E_{v}[\bar{\phi}(v_i)\bar{x}_i(v_i)] - \sum_{j=1}^h (G(F(v_i)) - H(F(v_i)))\bar{x}_i(v_i)|_{v_i =a_j}^{b_j} \\ 
 & \hspace{0.5cm} + \sum_{j=1}^h \int_{a_j}^{b_j} (H(v_i) - G(v_i)) d\bar{x}_i(v_i) \text{\tag*{(integration by parts)}}\\
  &= E_{v}[\bar{\phi}(v_i)\bar{x}_i(v_i)] + \sum_{j=1}^h \int_{a_j}^{b_j} (H(F(v_i)) - G(F(v_i))) d\bar{x}_i(v_i)
\end{align*}
}

The last equality holds since $G(0) = H(0)$ and $G(1) = H(1)$ by the definition of $G$ and $b_i = a_{i+1}$ for all $1 \leq i \leq h - 1$. In the second term of the last line, the derivative of $\bar{x}_i$ is non-positive, since $\bar{x}_i(v_i)$ is decreasing in $v_i$. Moreover, $H(F(v_i)) - G(F(v_i))$ is non-negative for all $v_i$, because $G(q) \leq H(q)$ for all $q \in [0, 1]$. In order to minimize the payment, we need to choose an allocation function $\bar{x}_i$ to minimize the magnitude of the second term. We show that the second term is zero when the mechanism buys from the $w$ players with smallest ironed virtual cost and breaks ties uniformly at random.

For any $q \in [0, 1]$, if $H(q) - G(q)$ is zero, then the contribution to the second term is zero. Thus, we only need to consider where $G$ and $H$ differ. Since $G$ is the convex hull of $H$, whenever $G < H$, $G$ must be flat. That is, for any $[a, b) \in S$, $g(q)$ has the same value for all $q \in [a, b)$. Since $\bar{\phi}(F^{-1}(q)) = g(q)$, every $v_i \in [F^{-1}(a), F^{-1}(b))$ has the same ironed virtual cost. Since the mechanism breaks ties uniformly at random, $\bar{x}_i(v_i)$ is constant for all $v_i \in [F^{-1}(a), F^{-1}(b))$. Thus, the derivative of $\bar{x}_i(v_i)$ is zero for all $v_i \in [F^{-1}(a), F^{-1}(b))$. Since $\bar{x}_i(v_i)$ is zero for all $v_i \in [F^{-1}(a), F^{-1}(b))$, it contributes nothing to the second term. Thus, the second term is always zero since if $H(F(v_i)) - G(F(v_i))$ is non-zero, then $\bar{x}_i(v_i)$ is zero. Hence, the mechanism minimizes the payment.
\end{proof}